\documentclass[letter,scriptaddress,twocolumn,showkeys]{revtex4-1}
\usepackage[latin9]{inputenc}

\usepackage{amsmath}
\usepackage{subfigure} 

\usepackage{tikz}
\usetikzlibrary{arrows}
\usetikzlibrary{positioning,automata}
\input{Qcircuit}

\newtheorem{theorem}{Theorem}
\newenvironment{proof}[1][Proof]{\noindent\textbf{#1.} }{\ \rule{0.5em}{0.5em}}
\newtheorem{corollary}[theorem]{Corollary}

\begin{document}

\title{Efficient implementation of Quantum circuits with limited qubit interactions}

\author{Stephen Brierley}

\affiliation{DAMTP, Centre for Mathematical Sciences, University of Cambridge, Wilberforce Road, Cambridge CB3 0WA, UK}

\begin{abstract}
The quantum circuit model allows gates between \emph{any} pair of qubits yet physical instantiations allow only limited interactions. We address this problem by providing an interaction graph together with an efficient method for compiling quantum circuits so that gates are applied only locally. The graph requires each qubit to interact with $4$ other qubits and yet the time-overhead for implementing any $n$-qubit quantum circuit is $6\log n$. Building a network of
quantum computing nodes according to this graph enables the network to emulate a single monolithic device with minimal overhead.
\end{abstract}

\maketitle

\section{Introduction}

Just as with their classical counterparts, quantum algorithms will
be compiled into a sequence of elementary physical operations.  Quantum algorithms use arbitrary two-qubit interactions since
in the circuit model, gates can be applied to \emph{any} pair of qubits.
However, after quantum error correction the allowed logical interactions
are limited to a graph that typically has low degree. Beals et al.
give a sequence of SWAP gates permuting the qubits so that every interaction
occurs between neighbours of the host graph \cite{beals+12}. The
time overhead, $T$, depends on the properties of the graph. Two interesting
examples being the $k$-dimensional lattice which for an $n$ qubit
device has overhead $T=O(n^{1/k})$ and the hypercube with overhead
$T=O(\log^{2}n)$ \cite{note1}. Comparing to the solution where each gate is implemented by a separate
permutation, this means that the time to permute all $n$ qubits is
within a logarithmic factor of the time to move just one.

The hypercube is a powerful network with the ability to sort in time
$O(\log^{2}n).$ However, the degree of each node grows as $\log n,$
which for large $n$ could become difficult to implement and means
that new components have to be designed as the device is scaled up.
In addition, implementations of optical switches in a noisy network
model typically suffer losses and so it is appealing to reduce the
degree to a small constant. In this paper we present improvements
to the approach taken by Beals et al. in two directions. We reduce
the required degree of the network to a small constant and at the
same time cut the overhead to $6\log n$ (see Table \ref{tab:results}
for a comparison to previous work). This lowers the cost of implementing
arbitrary quantum algorithms on a physical device and makes the required
networks more realistic. A device built using this architecture is
truly scalable, additional nodes have the same small degree as the
existing qubits. In addition, the lower degree means that we have
reduced the \emph{total} number of connections by a factor $O(\log n)$. 

\begin{table}[ht]
\begin{ruledtabular}
\begin{tabular}{lcccc}
Graph & Degree & $T$ & $S$ & Emulation method\tabularnewline
\hline 
Complete graph & $n$ & 1 & 1 & n/a\tabularnewline
1D n.-n. & 2 & $n^{2}$ & 1 & Move individually\tabularnewline
 &  & $2n-3$ & 1 & Sorting network \cite{beals+12,Hirata+11}\tabularnewline
 &  & $O(1)$ & $n$ & Teleportation \cite{Rosenbaum13}\tabularnewline
2D n.-n. & 4 & $O(\sqrt{n})$ & 1 & Sorting network \cite{beals+12}\tabularnewline
Hypercube & $\log n$ & $O(\log^{2}n)$ & 1 & Sorting network \cite{beals+12}\tabularnewline
 &  &  &  & \tabularnewline
Cyclic butterfly  & 4 & $6\log n$ & 2 & Theorem \ref{main theorem}\tabularnewline
\end{tabular}
\end{ruledtabular}
\caption{\label{tab:results} The time, $T$, and space, $S,$ overhead of embedding a quantum circuit into the graph restricted by the physical implementation. A key limitation being the degree of the graph which corresponds to number of interactions per qubit. 
Previous results have applied to the 1D and 2D nearest-neighbour (n.-n.) and hypercube graph. The final line summarizes the main result of this paper. We show that using a cyclic butterfly network reduces both the degree and time overhead in emulating a quantum circuit on a physically realistic device. }
\end{table}

In section \ref{sec:Hypercubic}, we introduce hypercube-like
networks and in particular, the so called cyclic butterfly network.
We then discuss the properties of a cyclic butterfly graph that we
need for the main result which is presented in Section \ref{sec:theorem}.
Some alternative networks and the application of these ideas to near-term
experiments on noisy network architectures are discussed in the conclusion.

\section{Hypercubic networks \label{sec:Hypercubic}}

We represent a network of qubits as an undirected graph. Nodes correspond
to single qubits, or qubit plus a single ancilla, and edges correspond
to the allowed interactions. The problem of permuting qubits is then
similar to routing packets of information in a synchronous parallel
computer. SWAP gates exchange quantum information between two nodes
or move a quantum state into a node provided there is an available
ancilla qubit in the state $|0\rangle$. In comparison to parallel
classical computing, the parameters we are interested in are somewhat
different. For example, we will think of each node as a \emph{single}
(or pair of) qubit(s) rather than a computing node capable of complex
operations. We clearly distinguish between the off-line classical
computation which is essentially free (provided it is poly-time) from
the on-line quantum computing. We also impose the restriction that
no two 'packets' can be stored at a single node; there is no 'buffering'
space in a single qubit. 

The quantum computer is required to work synchronously at the logically
level - of course at the physical scale, entanglement generation or
magic state distillation will be probabilistic and gate times will
vary. We do not address these issues here but rather assume that sufficient
physical resources allow the system to effectively function as a synchronous
device.

The computational power of a network is typically described in terms
of its ability to emulate the complete graph. Hypercubic networks
are variants of the hypercube that are designed to use nodes with
constant degree yet maintain its computational power to within a small
constant. Since we consider each node as a qubit, the low degree means
that we do not require too many possible interactions with other qubits.
In addition, hypercubic networks typically have a nice scaling property
since we can use the same components in any size quantum computer
(although the distance of the interactions may grow). There are many
hypercubic networks with prominent examples being the butterfly, cube-connected
cycles, Benes network, shuffle-exchange and the de Bruign network
(see for example, ref. \cite{Leighton}). We will use the so called
cyclic butterfly network (defined below) which has two useful properties;
it embeds a Benes network and is invariant under cyclic permutations.

\subsection{The cyclic butterfly network \label{sub:cyclic-butterfly}}

The $n=r2^{r}$ nodes of an $r$-dimensional cyclic butterfly network
(also called a wrapped butterfly) can be described in terms of the
rows and columns of an $r\times2^{r}$ array. Each node is labelled
by a pair $(w,i)$ where $w$ is a $r$-bit word corresponding to
one of the $2^{r}$ rows and $i$ labels the column. Two nodes $(w,i)$
and $(v,i+1\mod r)$ are connected by an edge if either they are in
the same row, $w=v$ or if $w$ and $v$ differ by precisely one bit
in position $i$. There are no other connections in the network so
the degree of every node equals $4$. An example of a $n=3\times2^{3}$
node cyclic butterfly network is given in Fig. 1.

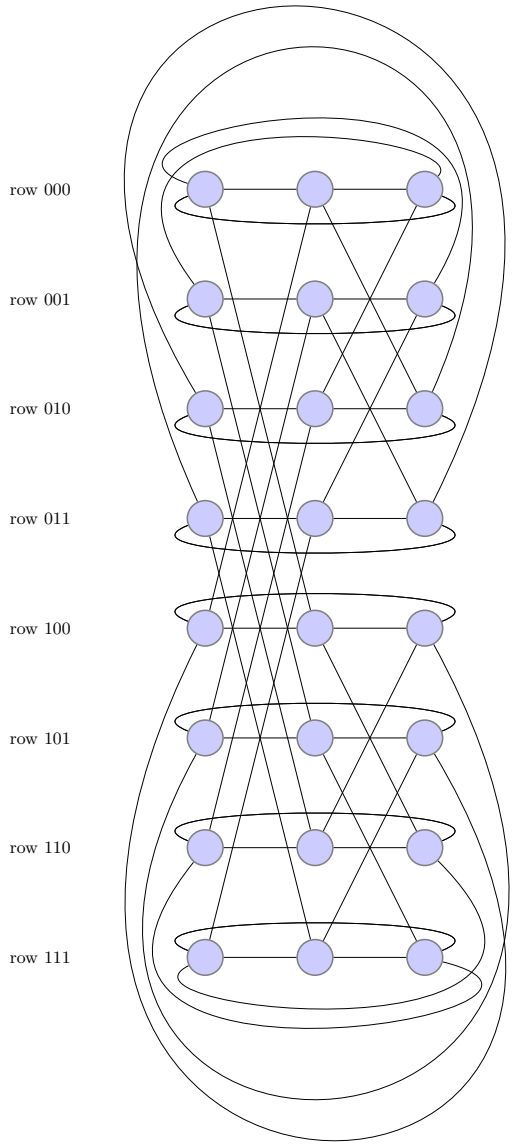
\begin{figure}[ht]
\center
\scalebox{0.73}{
\begin{tikzpicture}[auto, main_node/.style={circle,fill=blue!20, draw=black!50,thick, minimum size=0.65cm}]
  
\foreach \row in {0,2,4,6,8,10,12,14}{
   \foreach \place/\name in {{(0,\row)/a\row},{(2,\row)/b\row},{(4,\row)/c\row}}{     
      \node[main_node] (\name) at \place {}; 
    }
}

\foreach \row in {0,2,4,6}{
   \foreach \source/\dest in {a\row/b\row,b\row/c\row}{   
      \path (\source) edge (\dest);
      \draw (a\row) .. controls (-2,0.8+\row) and (6,0.8+\row) .. (c\row);
   }
}

\foreach \row in {8,10,12,14}{
   \foreach \source/\dest in {a\row/b\row,b\row/c\row}{   
      \path (\source) edge (\dest);
      \draw (a\row) .. controls (-2,-0.8+\row) and (6,-0.8+\row) .. (c\row);
   }
}

\foreach \source/\dest in {a0/b8,a2/b10,a4/b12,a6/b14}{   
      \path (\source) edge (\dest);
}

\foreach \source/\dest in {a8/b0,a10/b2,a12/b4,a14/b6}{   
      \path (\source) edge (\dest);
}

\foreach \source/\dest in {b0/c4,b2/c6,b4/c0,b6/c2}{   
      \path (\source) edge (\dest);
}

\foreach \source/\dest in {b8/c12,b10/c14,b12/c8,b14/c10}{   
      \path (\source) edge (\dest);
}

\draw (a0) .. controls (-2,-1) and (8,-2) .. (c2);
\draw (a2) .. controls (-4,-3) and (8,-1) .. (c0);

\draw (a4) .. controls (-5,-5) and (10,-5) .. (c6);
\draw (a6) .. controls (-6,-6) and (10,-6) .. (c4);

\draw (a8) .. controls (-5,19) and (8,19) .. (c10);
\draw (a10) .. controls (-6,20) and (10,20) .. (c8);

\draw (a12) .. controls (-3,16) and (5,15) .. (c14);
\draw (a14) .. controls (-3,15) and (7,17) .. (c12);

\node[] at (-3,0) {row 111};
\node[] at (-3,2) {row 110};
\node[] at (-3,4) {row 101};
\node[] at (-3,6) {row 100};
\node[] at (-3,8) {row 011};
\node[] at (-3,10) {row 010};
\node[] at (-3,12) {row 001};
\node[] at (-3,14) {row 000};

\end{tikzpicture}
}

\protect\caption{A 3-dimensional cyclic butterfly graph with $n=3\times2^{3}$ nodes
representing a qubit plus its ancilla. The edges represent the allowed
interactions between qubits. }
\end{figure}

The cyclic butterfly network is closely related to the hypercube.
Merging the $r$ nodes in every row into a single node results in
the $2^{r}$ node hypercube. Like the hypercube, the butterfly network
has a simple recursive structure, one $r$-dimensional butterfly contains
two $(r-1)$-dimensional butterflies. 

There are two properties of cyclic butterfly networks that we make
use of in our efficient algorithm for moving qubits. The first property
is that the graph embeds a so called Benes network \cite{Benes65},
meaning that if we traverse the graph with column label increasing
from $i=0\rightarrow r\equiv0$ and then back, $i=r\rightarrow0,$
we can implement any permutation of the $w=0\ldots2^{r}-1$ row labels
without collisions. The second property is that the graph is cyclic:
reordering the rows $i\mapsto i+1\mod r$ results in the same cyclic
butterfly graph. Combining these two properties means that every column
can traverse a Benes network simultaneously. Thus on a cyclic butterfly,
we can permute the $2^{r}$ row elements in every column without collisions.
Note that this is trivially true on a square $\sqrt{n}\times\sqrt{n}$
lattice: we can simultaneously permute the $\sqrt{n}$ entries of
every column independently. The crucial difference is that on a cyclic
butterfly the time taken is only $2r\approx2\log n$ as opposed to
$\sqrt{n}$ on a square lattice.

\section{Algorithm for permuting qubits \label{sec:theorem}}

We now present the main result of the paper, that the butterfly network
can implement any quantum algorithm with an overhead of $6\log n.$

\begin{theorem} \label{main theorem}On a $n$-qubit cyclic butterfly
network, there is a sequence of local gates with depth $6\log n$
such that the qubit at node $a$ is sent to node $\pi(a)$ for all
$a=1,\ldots,n$ and any permutation $\pi:[1,n]\rightarrow[1,n].$
\end{theorem}

\begin{proof} We use the row and column structure of the graph. The
destinations of every qubit are label-ed by $2^{r}$ rows, $w$, and
$r$ columns indexed by $i=0,\ldots,r-1$. We implement a permutation
of all nodes in three steps using this structure: we first permute
the rows, then columns and finally the rows again. The only moves
we are allowed to make is swapping two qubits or moving a qubit from
one node into its neighbours ancilla. In particular, no two qubits
can occupy the same node in a single step.

We first permute the entries in each row in such a way that the row
destination of every qubit in each column become distinct i.e. after
permuting rows, column $i,$ contains every word $w=0,\ldots,2^{r}-1$
for all $i=0,\ldots,r-1$. This is made possible by Hall's Matching
Theorem \cite{Hall35} - also called Hall's marriage theorem as it
allows two groups of men and women to happily marry. A matching in
a graph is a set of edges that have no common vertices. Hall's theorem
gives a necessary and sufficient condition for finding a matching
and is commonly used in routing problems.

We use the permutation $\pi$ to construct a bipartite ``routing
graph\textquotedblright{} $(U,V,E)$ containing $22^{r}$ nodes $U=\{u_{1},\ldots,u_{2^{r}}\}$
and $V=\{v_{1},\ldots,v_{2^{r}}\}$ and $r2^{r}$ edges $U=\{e_{1},\ldots,e_{r2^{r}}\}.$
The $U$ nodes represent the original row location of each qubit and
the $V$ nodes are their destination rows. If a qubit in row $u_{i}$
has a destination row $v_{j}$ we add the edge $(u_{i},v_{j})$ so
that there are $r$ edges for every node in $U$ and $V$. 

Hall's Matching Theorem then tells us that we can $r-$colour the
edges so that no colour is used twice at any node. We can use the
Ford-Fulkerson algorithm to find the matching by reducing the problem
to a maximum-flow problem \cite{Cormen+09}. We add two nodes $s$
and $t$ to the graph and connect $s$ to everything in $U$ and $t$
to everything in $V$. Since each node has unit capacity, a matching
is equivalent to the maximum flow from $s$ to $t$. The classical
computation of the Ford-Fulkerson algorithm is bounded by $O(|U||E|)=O(n^{2})$
\cite{Ford+56}. Having coloured the edges, we now know how to permute
the row elements; an operation we can implement in time $2r-3$ using
an insertion sorting network since each row is a 1D nearest neighbour
graph (see Appendix). 

The $r-$colouring implies that in every column, $i$, each row label
appears exactly once. Using the Benes and pipe-lining properties of
the butterfly network discussed in Sec \ref{sub:cyclic-butterfly},
we can sort every column according to the row labels in $2r$ steps.
In the first $r$ steps, the qubits increment $i\mapsto i+1\mod r$,
then in the final $r$ step the rows move in the opposite direction
$i\mapsto i-1\mod r$. Using a single ancilla at each node the time
cost is $2r$.

The final part of the algorithm is to permute the rows according to
the column labels. Since the destination column labels are now all
distinct, this is possible without collisions using insertion sort.

The total time overhead is thus $T=(2r-3)+(2r)+(2r-3)<6\log n$ as
claimed.\end{proof}

\begin{corollary} \label{corollary} A quantum computer whose $n$
logical qubits are connected according to the cyclic butterfly network
can implement any quantum algorithm with a time and space overhead
of $T=6\log n$ and $S=2$ respectively. \end{corollary}

\begin{proof} Each time-step in a quantum circuit consists of up
to $n/2$ two-qubit gates. The gates define the permutation, $\pi$,
used in Theorem 1. We place the destination of each pair of qubits
involved in a gate so that they are neighbours in the cyclic-butterfly
graph. The proof of Theorem 1 provides an efficient method to construct
a sequence of gates implementing the permutation. Every time step
requires one permutation of the qubits so the time and space overhead
is precisely that given in Theorem 1. \end{proof}

\section{Conclusion}

Quantum computers are fully parallel machines. Every qubit is effectively
a processing node since the identity gate will be error corrected
at a cost similar to other gates. Taking this view has led to the
application of techniques developed for routing in synchronous parallel
(classical) computers. We presented an efficient method for compiling
a quantum circuit onto a cyclic-butterfly network. This improves on
previous results in two respects. The interaction graph has constant
degree and at the same time, the time overhead is a small constant
away from the best possible (the time to move a single qubit). 

There are two alterations to the cyclic butterfly graph one could
make that achieve a trade-off between the cost of building the network and the time-overhead in emulating arbitrary circuits. 
\begin{enumerate}
\item
Replace each node by a ring of 4 nodes, each connected
to one of the previous edges. This reduces the connectivity to 3,
the minimum possible non-trivial degree, whilst increasing the time
overhead by a factor 2. 

\item
Use the $k$-arry cyclic butterfly graph. In this case, the degree increase to $2k$
whilst reducing the overhead to $T=6\log_{k}n.$ 
\end{enumerate}
Combining these two
ideas results in a slightly more efficient solution than the cyclic
butterfly graph. The $k$-arry cyclic butterfly with each node expanded
to a ring of $2k$ nodes has degree 3 and time overhead $T=6k\log_{k}n$,
thus taking $k=3$ is optimal.

The ideas presented here can used when designing the communication
architecture in a noisy network quantum computer. Individual nodes
(or cells) correspond to a small number of physical qubits in a system
such as NV centers in diamond, trapped ions or superconducting devices.
Photonic channels mediate entanglement between two nodes which can
then be distilled to allow inter-node communication (see, for example,
recent experimental results in NV centers \cite{Bernien+13}, superconducting
qubits \cite{Roch+14} and trapped ions \cite{Hucul+15}). Nickerson
et al. show how these resources could be used to implement a fault
tolerant computation via the surface code even in the presence of
noisy photonic links \cite{Nickerson+14}. An alternative approach
would be to take advantage of the cyclic butterfly graph and use CSS
block codes. Steane described how fault tolerant operations can be
performed on separate CSS block codes via ancilla states \cite{Steane99,Brun+15}.
Thus nodes could correspond to a small number of logical qubits, each
in a separate block. The ancilla states would then be distilled using
the photonic channel in much the same way as 4-qubit GHZ states are
required when using the surface code.

\section*{Acknowledgments}

The author would like to thank Aram Harrow and Naomi Nickerson for
suggesting the two alternative graphs given in the conclusion, and
Tom McCourt for discussions when the idea was at an early stage.

\section*{Appendix: Sorting Networks}

A sorting network is designed to sort all possible input sequences
using only comparison gates acting on neighbouring nodes $(x,y)\in G$,
\[
C(x,y)=\begin{cases}
(x,y) & \mbox{if }x>y\\
(y,x) & \mbox{if }x<y.
\end{cases}
\]
That is, $C(x,y)$ swaps the inputs if $x<y$ and leaves them unchanged
otherwise. Sorting networks have been well studied in the classical
literature and examples are know over various graphs \cite{Knuth}.
Two examples are insertion sort and bitonic sort that sort over the
1D nearest-neighbour and hypercubic graphs respectively (see Fig \ref{fig:sorting networks}).
With full parallelism, bubble sort and insertion sort lead to the
same 1D nearest neighbour algorithm and require time $T=2n-3$.

A sorting network over a graph, $G$, provides a method of compiling
any circuit onto $G$. Each time-step in the original circuit defines
a permutation; qubits are moved so that the gates become local in
$G$. The classical compiler then inputs the destinations into the
sorting network and each time the comparison gate implements a SWAP,
the compiler applies a SWAP gate to the corresponding qubits. By construction,
every operation is local in $G$ and once the required gates from
the sorting network have been added, the gates from the time-step
in the original circuit can be enacted on neighbouring qubits. Note
that it is not necessary to have a sorting network that correctly
sorts all inputs, we only need to sort the inputs that appear in the
circuit. In addition, one could use a different network for each time-step.

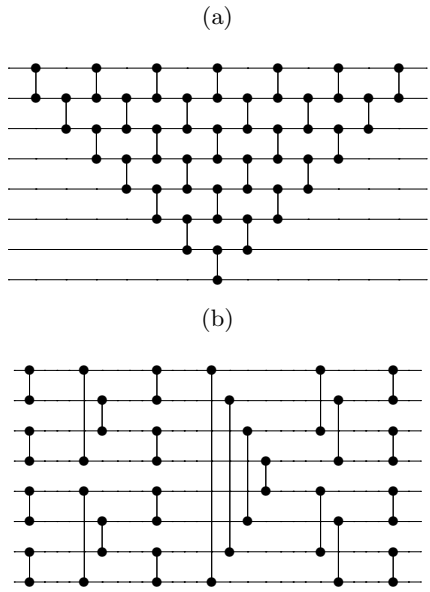
\begin{figure}[ht]
\centering 
(a)
\[ \Qcircuit @R=1em @C1em { 
\lstick{} & \ctrl{1}  & \qw & \ctrl{1} & \qw & \ctrl{1} & \qw  & \ctrl{1} & \qw & \ctrl{1} & \qw  & \ctrl{1} & \qw & \ctrl{1} & \qw\\
\lstick{} & \ctrl{-1} & \ctrl{1}  & \ctrl{-1} & \ctrl{1} & \ctrl{-1} & \ctrl{1} & \ctrl{-1}  & \ctrl{1} & \ctrl{-1} & \ctrl{1} & \ctrl{-1}  & \ctrl{1} & \ctrl{-1} & \qw \\
\lstick{} & \qw & \ctrl{-1} & \ctrl{1}  & \ctrl{-1} & \ctrl{1} & \ctrl{-1} & \ctrl{1} & \ctrl{-1}  & \ctrl{1} & \ctrl{-1} & \ctrl{1} & \ctrl{-1}  & \qw & \qw \\
\lstick{} & \qw & \qw  & \ctrl{-1} & \ctrl{1} & \ctrl{-1} & \ctrl{1} & \ctrl{-1}  & \ctrl{1} & \ctrl{-1} & \ctrl{1} & \ctrl{-1} & \qw & \qw & \qw \\
\lstick{} & \qw & \qw  &\qw & \ctrl{-1} & \ctrl{1} & \ctrl{-1} & \ctrl{1}  & \ctrl{-1} & \ctrl{1} & \ctrl{-1} & \qw & \qw& \qw & \qw \\
\lstick{} & \qw & \qw  & \qw & \qw & \ctrl{-1} & \ctrl{1} & \ctrl{-1}  & \ctrl{1} & \ctrl{-1}& \qw & \qw & \qw& \qw & \qw \\
\lstick{} & \qw & \qw  & \qw & \qw & \qw & \ctrl{-1} & \ctrl{1}  & \ctrl{-1} & \qw & \qw & \qw & \qw& \qw & \qw \\
\lstick{} & \qw & \qw  & \qw & \qw & \qw & \qw & \ctrl{-1}  & \qw & \qw & \qw & \qw & \qw& \qw & \qw \\
} \]

(b)
\[ 
\Qcircuit @R=1em @C0.5em { 
\lstick{} & \ctrl{1}  & \qw & \qw & \qw& \ctrl{3} & \qw  & \qw & \qw & \qw& \ctrl{1}  & \qw & \qw & \qw& \ctrl{7} & \qw & \qw & \qw  & \qw & \qw & \qw& \ctrl{2} & \qw  & \qw & \qw & \qw& \ctrl{1}  & \qw & \qw \\ 
\lstick{} & \ctrl{-1}  & \qw & \qw & \qw& \qw & \ctrl{1}  & \qw & \qw & \qw& \ctrl{-1}  & \qw & \qw & \qw& \qw & \ctrl{5} & \qw & \qw  & \qw & \qw & \qw& \qw & \ctrl{2}  & \qw & \qw & \qw& \ctrl{-1}  & \qw & \qw \\ 
\lstick{} & \ctrl{1}  & \qw & \qw & \qw& \qw & \ctrl{-1}  & \qw & \qw & \qw& \ctrl{1}  & \qw & \qw & \qw& \qw & \qw & \ctrl{3} & \qw  & \qw & \qw & \qw& \ctrl{-2} & \qw  & \qw & \qw & \qw& \ctrl{1}  & \qw & \qw \\ 
\lstick{} & \ctrl{-1}  & \qw & \qw & \qw& \ctrl{-3} & \qw  & \qw & \qw & \qw& \ctrl{-1}  & \qw & \qw & \qw& \qw & \qw & \qw & \ctrl{1}  & \qw & \qw & \qw& \qw & \ctrl{-2}  & \qw & \qw & \qw& \ctrl{-1}  & \qw & \qw \\ 
\lstick{} & \ctrl{1}  & \qw & \qw & \qw& \ctrl{3} & \qw  & \qw & \qw & \qw& \ctrl{1}  & \qw & \qw & \qw& \qw & \qw & \qw & \ctrl{-1}  & \qw & \qw & \qw& \ctrl{2} & \qw  & \qw & \qw & \qw& \ctrl{1}  & \qw & \qw \\ 
\lstick{} & \ctrl{-1}  & \qw & \qw & \qw& \qw & \ctrl{1}  & \qw & \qw & \qw& \ctrl{-1}  & \qw & \qw & \qw& \qw & \qw & \ctrl{-3} & \qw  & \qw & \qw & \qw& \qw & \ctrl{2}  & \qw & \qw & \qw& \ctrl{-1}  & \qw & \qw \\ 
\lstick{} & \ctrl{1}  & \qw & \qw & \qw& \qw & \ctrl{-1}  & \qw & \qw & \qw& \ctrl{1}  & \qw & \qw & \qw& \qw & \ctrl{-5} & \qw & \qw  & \qw & \qw & \qw& \ctrl{-2} & \qw  & \qw & \qw & \qw& \ctrl{1}  & \qw & \qw \\ 
\lstick{} & \ctrl{-1}  & \qw & \qw & \qw& \ctrl{-3} & \qw  & \qw & \qw & \qw& \ctrl{-1}  & \qw & \qw & \qw& \ctrl{-7} & \qw & \qw & \qw  & \qw & \qw & \qw& \qw & \ctrl{-2}  & \qw & \qw & \qw& \ctrl{-1}  & \qw & \qw \\ 
} 
\]

\protect\caption{\label{fig:sorting networks}Two examples of sorting networks on 8
inputs: (a) the insertion sort over a 1D nearest neighbour graph which
sorts in time $T=2n-3$, and (b) the bitonic sort over the hypercube
that requires time $T=\tfrac{1}{2}\log n(\log n+1)$.}
\end{figure}


\begin{thebibliography}{10}

\bibitem{beals+12} R. Beals, S. Brierley, O. Gray, A. Harrow, S.
Kutin, N. Linden, D. Shepherd and M. Stather, \textit{Efficient Distributed
Quantum Computing}, Proc. R. Soc. A 2013 469, 20120686. arXiv:1207.2307

\bibitem{note1} There is an asymptotically better algorithm with overhead $O(\log n)$ on the hypercube. However, we will not consider it here since it is
based on the AKS sorting algorithm \cite{AKS,Paterson90} which has constant $\approx6,100$ \cite{Knuth}.

\bibitem{AKS} M. Ajtai, J. Komlos and E. Szemeredi, \emph{An $O(n\log n)$
sorting network,} Proc. 15th annual ACM symposium on Theory of computing,
1 (1983)

\bibitem{Paterson90} M. Paterson, \emph{Improved sorting networks
withO (logN) depth}, Algorithmica 5 (1-4), 75-92, (1990)

\bibitem{Knuth} D. Knuth, \emph{The Art of Computer Programming,
Volume 3: Sorting and Searching, }Addison-Wesley, 1998

\bibitem{Hirata+11} Y. Hirata, M. Nakanishi, S. Yamashita and Y Nakashima,
\emph{An efficient conversion of quantum circuits to a linear nearest
neighbor architecture,} Quantum Information \& Computation 11, 142
(2011)

\bibitem{Rosenbaum13} D. Rosenbaum, \emph{Optimal Quantum Circuits
for Nearest-Neighbor Architectures}, pg 294, 8th Conference on the
Theory of Quantum Computation, Communication and Cryptography (TQC
2013), Schloss Dagstuhl--Leibniz-Zentrum fuer Informatik. arXiv:1205.0036

\bibitem{Leighton} F. Leighton, \emph{Introduction to parallel algorithms
and architectures,} Morgan Kauffman Publishers, San Mateo, CA. (1992)

\bibitem{Benes65} V. Benes, \emph{Mathematical theory of connecting
networks and telephone traffic}, Academic Press Inc., New Yrok, 1965.

\bibitem{Hall35} P. Hall, Philip, \emph{On Representatives of Subsets},
J. London Math. Soc. 10 (1): 26\textendash 30, (1935)

\bibitem{Bernien+13} H. Bernien, B. Hensen, W. Pfaff, G. Koolstra,
M. S. Blok, L. Robledo, T. H. Taminiau, M. Markham, D. J. Twitchen,
L. Childress and R. Hanson, \emph{Heralded entanglement between solid-state
qubits separated by three metres,} Nature 497, 86\textendash 90 (2013)

\bibitem{Roch+14} N. Roch, M. E. Schwartz, F. Motzoi,
C. Macklin, R. Vijay, A. W. Eddins, A. N.
Korotkov, K. B. Whaley, M. Sarovar, and I. Siddiqi, \emph{Observation
of Measurement-Induced Entanglement and Quantum Trajectories of Remote
Superconducting Qubits,} Phys. Rev. Lett. 112, 170501 (2014)

\bibitem{Hucul+15} D. Hucul, I. Inlek, G. Vittorini, C. Crocker,
S. Debnath, S. Clark and C. Monroe, \emph{Modular entanglement of
atomic qubits using photons and phonons, }Nature Physics 11, 37\textendash 42
(2015) 

\bibitem{Nickerson+14}N. Nickerson, J. Fitzsimons and S. Benjamin,
\emph{Freely scalable quantum technologies using cells of 5-to-50
qubits with very lossy and noisy photonic links}, Phys. Rev. X 4,
041041. arXiv:1406.0880

\bibitem{Steane99} A. Steane, \emph{Efficient fault-tolerant quantum
computing,} Nature 399, 124-126 (1999). arXiv:quant-ph/9809054

\bibitem{Brun+15} T. Brun, Y.-C. Zheng, K.-C. Hsu, J. Job and C.-Y.
Lai, \emph{Teleportation-based Fault-tolerant Quantum Computation
in Multi-qubit Large Block Codes}. arXiv:1504.03913

\bibitem{Cormen+09} T Cormen, C Leiserson, R Rivest and C Stein,
\emph{Introduction to Algorithms}, MIT Press Cambridge Massachusetts
2009

\bibitem{Ford+56}L. Ford and D. Fulkerson, \emph{Maximal flow through
a network}, Canadian Journal of Mathematics 8: 399 (1956)\end{thebibliography}
\end{document}